\documentclass[runningheads]{llncs}
\usepackage{amsmath}
\usepackage{amssymb}
\usepackage{mathtools}
\usepackage{tikz}

\usepackage{tikz-cd}
\usetikzlibrary{automata, positioning, arrows}
\tikzset{->, node distance=2cm, every state/.style={thick, fill=gray!10}, initial text=$$,
}
\usepackage{nicefrac}
\usepackage{pict2e}
\usepackage{enumitem}
\usepackage{stmaryrd}

\newcommand*{\Defeq}{\coloneqq}
\newcommand{\Def}{\Defeq}

\newcommand{\COMMENT}[1]{}

 \newcommand{\Filter}[1]{\,\triangleright\,}

\newcommand{\F}[1]{{\mathbb F}_{#1}}
\newcommand*{\Nat}{\mathbb{N}} 
\newcommand*{\Int}{\mathbb{Z}}

\newcommand*{\Real}{\mathbb{R}}

\newcommand*{\Q}{\mathbb{Q}}

\newcommand*{\Inter}{\,\cap\,}
\newcommand*{\Union}{\,\cup\,}

\newcommand*{\Max}[2]{\mathrm{max}({#1},{#2})}

\newcommand*{\Inverse}[1]{{#1}^{-1}}

\newcommand*{\comment}[1]{}

\renewcommand*{\And}{\,\land\,}
\newcommand*{\Or}{\,\lor\,}
\newcommand*{\Imp}{\,\Rightarrow\,}

\newcommand*{\Implies}{\Imp}
\newcommand*{\Iff}{\,\iff\,}
\newcommand*{\Exists}[2]{(\exists {#1})\,{#2}}
\newcommand*{\Forall}[2]{(\forall {#1})\,{#2}}

\newcommand*{\Cross}{\times}

\newcommand{\Pairing}[2]{\langle {#1},\,{#2}\rangle}

\renewcommand{\S}{{\mathcal{S}}}        
\newcommand{\R}{{\mathcal{R}}}          

\newcommand{\K}{{\mathcal{K}}}          

\newcommand{\Tl}[1]{{#1}'}

\newcommand{\iterTl}[2]{{#1}^{({#2})}}

\newcommand{\LI}[1]{\mathrm{I}({#1})} 

\newcommand{\PosCone}[1]{{#1}_{+}}  

\newcommand{\LOR}{{\mathcal{L}}_{or}}            
\newcommand{\TORCF}{\mathcal{T}_{\text{rcf}}}   
\newcommand{\RCVR}{\mathcal{T}_{\text{rcvr}}}   
\newcommand{\TORCR}{\RCVR}

\newcommand{\Poly}[2]{{#1}[{#2}]}
\newcommand{\Rational}[2]{{#1}({#2})}
\newcommand{\Power}[2]{{#1}\llbracket {#2}\rrbracket}
\newcommand{\Laurent}[2]{{#1}(\!({#2})\!)}  

\newcommand{\Span}[1]{\left\langle{#1}\right\rangle}

\newcommand{\lbparen}{%
  \mathopen{%
    \sbox0{$()$}%
    \setlength{\unitlength}{\dimexpr\ht0+\dp0}%
    \raisebox{-\dp0}{%
      \begin{picture}(.32,1)
      \linethickness{\fontdimen8\textfont3}
      \roundcap
      \put(0,0){\raisebox{\depth}{$($}}
      \polyline(0.32,0)(0,0)(0,1)(0.32,1)
      \end{picture}%
    }%
  }%
}

\newcommand{\rbparen}{%
  \mathclose{%
    \sbox0{$()$}%
    \setlength{\unitlength}{\dimexpr\ht0+\dp0}%
    \raisebox{-\dp0}{%
      \begin{picture}(.32,1)
      \linethickness{\fontdimen8\textfont3}
      \roundcap
      \put(-0.08,0){\raisebox{\depth}{$)$}}
      \polyline(0,0)(0.32,0)(0.32,1)(0,1)
      \end{picture}%
    }%
  }%
}

\newcommand{\Ana}[1]{\lbparen\,{#1}\,\rbparen}

\begin{document}
\title{A Decision Method for \\ Elementary Stream Calculus}
%
%
\author{Harald Ruess\inst{1}
}
\authorrunning{H. Ruess}

%
\institute{Entalus Computer Science Labs 
            \\[2mm]
            \email{harald.ruess@entalus.com}\\[4mm]
            Preprint\\
            $10^{\mathit{th}}$ December 2023
}
\maketitle              
\begin{abstract}
The main result is a doubly exponential decision procedure for the first-order equality theory of streams with both arithmetic and control-oriented stream operations.
This stream logic is expressive for elementary problems of stream calculus.

\end{abstract}

\section{Introduction}\label{sec:intro}
The principled design of stream-based systems in computer science and control theory heavily relies on solving quantified stream constraints~\cite{gilles1974semantics,kahn1976coroutines,burge1975stream,broy2012specification,broy1986theory,broy2023}\@.
Solving these stream constraints, however, is challenging as quantification over streams is effectively second order.

The monadic second-order logic $MSO(\omega)$~\cite{courcelle2012graph,owre2000integrating} over $\omega$-infinite words, in particular, is an expressive logic for encoding quantified constraints 
over discrete streams by quantification over sets of natural numbers\@. 
The set of models of any $MSO(\omega)$ formula
can be characterized by a finite-state Mealy machine~\cite{buchi1969definability}\@.
This correspondence forms the basis of a non-elementary decision procedure for
$MSO(\omega)$,
since emptiness for finite-state 
Mealy machines is decidable\@. 
Equivalently, the first-order equality theory of streams is
non-elementarily decidable based on the logic-automaton correspondence~\cite{pradic2020some}\@.

{\em Main result.} For a given first-order formula $\varphi$ 
in the language of ordered rings with equality, the validity of $\varphi$ in the 
structure of discrete, real-valued streams can be decided using quantifier elimination 
in doubly exponential time.
Definitional extensions demonstrate the expressive power of stream logic for deciding elementary problems in stream calculus~\cite{rutten2019method}\@.

In contrast to automata-based procedures for second-order monadic logics, our decision procedure is not restricted to finite alphabets, it has doubly exponential time complexity, and it is amenable to various conservative extensions for encoding typical correctness conditions in system design.
In particular, stream logic is expressive in supporting both arithmetic and more control-oriented operations such as finite stream shifting. 
Such a combination
of arithmetic with control-oriented operations has also been a long-standing challenge in the design of decision procedures for finite streams~\cite{cyrluk1997efficient,moller1998solving}\@.

Our developments are structured as follows.
In Section~\ref{sec:examples} we provide some motivating examples of typical stream constraints 
in {\em stream calculus}~\cite{rutten2019method}\@. 
Section~\ref{sec:streams} summarizes essential facts on streams based on 
their interpretation as formal power series~\cite{niven1969formal}, with the intent of making this exposition largely self-contained\@.
Streams are identified with {\em formal power series} and the superset of streams with finite history prefixes are identified with {\em formal Laurent series}\@. 
Streams are orderable and they are Cauchy complete with respect to the prefix ultrametric.
We show in Section~\ref{sec:computable} that streams are a {\em real closed valuation ring} and their extensions with finite histories are a real-closed field. 
The main technical hurdle in these developments is the derivation of an {\em intermediate value property} (IVP) for streams, since streams, as  ordered and complete non-Archimedean domains,
are lacking the {\em least upper bound property}, and therefore the usual
dichotomic procedure for proving IVP does not apply\@.
As a consequence of the real-closedness of streams, the ordered ring (and field) of streams admit quantifier elimination.
The results in Section~\ref{sec:qe} are direct consequences of the quantifier elimination procedures for 
real-closed valuation rings~\cite{cherlin1983real} and
for real-closed ordered fields~\cite{tarski1998decision} together with the doubly exponential bound obtained by Collin's~\cite{collins1975quantifier} procedure in the case of real-closed ordered fields\@.
In Section~\ref{sec:extensions}\@,
we demonstrate the application of decidable stream logic to the 
analysis of stream circuits by conservative extensions for the shift operation, constants for rational and automatic streams, and stream projection.
We conclude with some final remarks in Section~\ref{sec:conclusions}\@.

\begin{figure}[t]

\begin{center}
    \begin{tikzcd}
    & h_1    \arrow[d, mapsto]
    & D_1        \arrow[l]
    & h_2    \arrow[l, mapsto]
    &
    \\
       z    \arrow[r, mapsto]
    & A  \arrow[r]
    & h_3 \arrow[r, mapsto]
    & C      \arrow[u]   
    &  y \arrow[l, leftarrow] 
   \end{tikzcd}
   \end{center}
    \caption{Finite Stream Circuit.}
    \label{fig:example.circuit}
\end{figure}
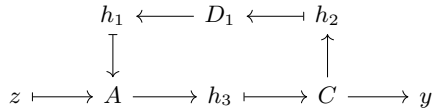

\section{Examples}\label{sec:examples}

We are demonstrating the r\^ole of first-order stream logic for encoding some elementary problems from stream calculus.

\paragraph{Observational Equivalence.}
Two stream processors $T_1$, $T_2$ are shown to be {\em observationally equivalent} by demonstrating the validity of the following first-order 
formula\@.  
\begin{example}[Observational Equivalence]\label{ex1}
  $$
       \Forall{x, y_1, y_2}
           T_1(x, y_1) \And T_2(x, y_2) \Implies y_1 = y_2 
    $$
\end{example}
Hereby, the logical variables $x$, $y_1$, and $y_2$ are interpreted over (discrete) real-valued streams, 
and $T_i(x, y_i)$, for $i = 1, 2$, are binary predicates for defining the possible output streams $y_i$ of processor $T_i$ on input stream $x$\@.

In {\em stream calculus}~\cite{rutten2019method}, the relations $T_i(x, y_i)$ are typically of the form $y_i = f \cdot x$, where the {\em transfer function} $f$ is a {\em rational stream}, and the output stream $y_i$ is obtained by stream convolution of $f$ with the input stream $x$\@.
These equality specifications are expressive for all well-formed stream circuits with loops~\cite{rutten2008rational}\@. 

\paragraph{Functionality.}
A stream processor $T$ is shown to be {\em functional} if the following first-order stream formula with one quantifier alternation holds. 
\begin{example}[Functionality]\label{ex2}
       $$(\forall x)(\exists y)\, T(x, y)
           \And (\forall z)\, z \neq y \Implies \lnot T(x, z) 
       $$
\end{example}
\paragraph{Non-Interference.}
We are now considering streams of system outputs, which are considered to be  partitioned into a {\em low} and a {\em high} security part. 
In these cases, a stream processor $T$ has a {\em non-interference}~\cite{goguen1982security,mccullough1988noninterference,mclean1994general} property if executing $T$
from indistinguishable low outputs results in 
indistinguishable low outputs at every step. 
\begin{example}[Non-Interference]\label{ex3}
       $$\Forall{x, y_1, y_1}{T(x, y_1) \And T(x, y_2) \Implies \mathit{hd}{(y_1)} =_L \mathit{hd}{(y_2)} \Implies y_1 =_L y_2} 
       $$
\end{example}
Hereby, $\mathit{hd}{(y_i)}$ for $i = 1, 2$ denote initial values, and
$y_1 =_L y_2$ is assumed to hold if and only if the {\em low} parts (for example, projections) of  $y_1$ and $y_2$ are equal.
These non-interference properties are prominent examples of a larger class of {\em hyperproperties}~\cite{clarkson2010hyperproperties} for comparing two or more system traces. 

\paragraph{Stream Circuits.}
We consider some typical design steps for the
stream circuit in 
Figure~\ref{fig:example.circuit}\@. 
At moment $0$ this circuit inputs the first value $z_0$ at its input end.
The present value $0$ of the register $D_1$ 
is added by $A$ to this and the result 
$y_0 = z_0 + 0 = z_0$ is the first value to be output.
At the same time, this value $z_0$ is copied (by $C$) and 
stored as the new value of the register $D_1$\@. 
The next step consists of inputting the value $z_1$, 
adding the present value of the register, $z_0$, 
to it and outputting the resulting value $y_1 = z_0 + z_1$\@.
At the same time, this value is copied and stored as the new value of the register.
The next step will input $z_2$ and output the value $y_2 = z_0 + z_1 + z_2$; and so on.
The stream transformer for the circuit in Figure~\ref{fig:example.circuit} is written as
   $y = (1, 1, 1, \ldots)\; \cdot\; z$\@,
whereby $'\cdot'$ denotes discrete stream convolution.
For verifying this statement about the stream circuit in 
Figure~\ref{fig:example.circuit} we might want to prompt a decision procedure for stream logic as follows. 
\begin{example}[Analysis]\label{ex4}
   \begin{align*}
   & (\forall z, y, h_1, h_2, h_3) \\ 
   & ~~~h_1 = D_1(h_2) \And
        h_3 = A(z, h_1) \And
        h_2 = C(h_3) \And
        y = h_3 \\
   & ~~~~~~ \Implies y = (1, 1, 1, \ldots) \cdot z
   \end{align*}
\end{example}
Hereby, $D_1(h_2) \Defeq X \cdot h_2$,\footnote{
$X$ might be viewed as an operation for right-shifting a stream by padding it with a leading $0$\@.
} $A(z, h_1) \Defeq z + h_1$, 
and $C(h_3) \Defeq h_3$, and the {\em rational streams} $X$ and
$(1, 1, \ldots)$ are considered to be constant symbols in the logic.
One may synthesize the transfer function
by constructing explicit witnesses for existentially quantified variables
in the underlying proof procedure. 
\begin{example}[Synthesis]\label{ex5}
   \begin{align*}
   & (\forall z, y, h_1, h_2, h_3) \\
   & ~~~h_1 = D_1(h_2) \And
        h_3 = A(z, h_1) \And
        h_2 = C(h_3) \And
        y = h_3 \\
   & ~~~~~~ \Implies (\exists u) y = u \cdot z
   \end{align*}
\end{example}

\section{Streams}\label{sec:streams}

Let $(\K, +, \cdot, \leq)$ be a totally ordered field. 
A $\K$-valued {\em stream} is an infinite sequence $(a_i)_{i \in \Nat}$ with  $a_i \in \K$\@.
Depending on the application context, streams are also referred to as discrete streams or signals, 
$\omega$-streams, $\omega$-sequences, or $\omega$-words\@.
The {\em generating function}~\cite{charalambides2018enumerative} of a stream is a
{\em formal power series}
   \begin{align}
   \sum_{i\in\Nat} a_i X^i
  \end{align}
in the {\em indefinite} $X$\@. 
These power series are {\em formal} as, at least in the algebraic view, 
the symbol $X$ is not being instantiated and there is no notion of convergence. 
We call $a_i$ the {\em coefficient} of $X^i$, and
the set of formal power series with coefficients in $\K$ is 
denoted by $\Power{\K}{X}$\@.
We also write $f_i$ for the coefficient of $X^i$ in the formal power series $f$\@. 
A {\em polynomial} in $\Poly{\K}{X}$ of degree $d \in \Nat$ is a formal power series $f$ which is {\em dull}, that is $f_d \neq 0$ and $f_i = 0$ for all $i > d$\@. 
For the one-to-one correspondence between streams and formal power series, we use these notions interchangeably. 
%
Now, addition of streams $f, g \in \Power{\K}{X}$ is pointwise, and streams are multiplied by means of discrete convolution.
   \begin{align}  \label{add}
       f + g           &\Defeq   \sum_{i\in\Nat}  (f_i + g_i) X^i \\  \label{multiply}
       f \cdot g          &\Defeq   \sum_{i \in\Nat}  (\sum_{j = 0}^i f_j g_{i-j}) X^i 
   \end{align}
$(\Power{\K}{X}, +, \cdot)$ is a principal ideal domain with the ideal $(X) = X\cdot\Power{\K}{X}$ the only non-zero maximal ideal.
The field $\K$ is embedded in the polynomial ring $\Poly{\K}{X}$, 
which itself is 
embedded in $\Power{\K}{X}$\@.
Moreover, the set of {\em rational functions} $\Rational{\K}{X}$ is defined as the fraction field of the polynomials $\Poly{\K}{X}$\@. 
Neither $\Power{\K}{X}$ nor $\Rational{\K}{X}$ contains the other. 
%
%
The {\em multiplicative inverse} $f^{-1}$, for $f \in \Power{\K}{X}$, exists if only if $f_0 \neq 0$\@.
We also write the quotient $\nicefrac{f}{g}$ instead of $f \cdot g^{-1}$\@.
\begin{example}[Rational Streams~\cite{rutten2008rational}]\label{ex:rats}
A stream in $\Power{\Real}{X}$ is {\em rational} if it can be expressed as a quotient $\nicefrac{p}{q}$ of polynomial streams $p, q \in \Poly{\Real}{X}$ such that $g_0 \neq 0$\@.
    \begin{align*}
    \nicefrac{1}{(1 - X)} &= (1, 1, 1, 1, \ldots) & \\
    \nicefrac{1}{(1 - X)^2} &= (1, 2, 3, 4, \ldots) & \\
    \nicefrac{1}{(1 - rX)} &= (1, r, r^2, r^3, \ldots) &\text{(for  $r \in \Real$)}\\
    \nicefrac{X}{(1 - X - X^2)} &= (0, 1, 1, 2, 3, 5, 8, \ldots) &\text{(Fibonacci)}  
    \end{align*}
The Fibonacci recurrence
$a_0 = 0$, $a_1 = 0$, and $a_{n} = a_{n-1} + a_{n-2}$ for all $n \geq 2$, for instance, is derived by equating corresponding coefficients on both sides of the equality 
       $ X 
     =  (1 - X - X^2) \cdot \sum_{k = 0}^\infty X^k
     =  \sum_{k = 0}^\infty (\sum_{j = 0}^k a_j) X^k$\@.
\end{example}
Rational streams are a subring of the formal power series $\Power{\Real}{X}$\@, and ultimately periodic 
streams such as $\omega$-words are a special 
case of rational streams \cite{rutten2008rational}\@.
Stream circuits as in Figure~\ref{fig:example.circuit} are expressive in representing any rational stream. 
\begin{example}[\cite{rutten2002streams}]~\label{ex:rational.stream.transformers}
Using  defining equations for $D_1$ (the unit delay register), $A$ (addition of two streams), and $C$ (copying of a stream)
we obtain 
from the stream circuit in Figure~\ref{fig:example.circuit} 
a corresponding system of defining equations 
     $
        h_1 = X \cdot h_2,\,
        h_3 = z + h_1,\,
        h_2 = h_3,\,
        y = h_3
     $\@.
Back substitution for the intermediate streams $h_3$, $h_1$, and $h_2$, in this order, yields an equational constraint 
     $y = z + (X \cdot y)$,
     which is equivalent to
     $ y = \nicefrac{1}{(1 - X)} \cdot z$\@.
Now, $y = (\sum_{i = 0}^k z_i)_{k\in \Nat}$
as a result of the identity for $\nicefrac{1}{(1 - X)}$ in Example~\ref{ex:rats}\@. 
\end{example}

\noindent
\begin{remark}
The ring of rational streams~\cite{rutten2008rational} is substantially different 
from the field $\Rational{\Real}{X}$ of {\em rational functions}\@.
For instance, the inverse $\nicefrac{1}{X}$ of the shift stream $X$ is not a rational stream, and 
it is not even a formal power series.
But it is a rational function in $\Rational{\Real}{X}$\@.
\end{remark}
The field $\Laurent{\K}{X}$ of {\em formal Laurent series} is the fraction field of the 
formal power series $\Power{\K}{X}$\@.
Elements of $\Laurent{\K}{X}$ are of the form
   \begin{align}
   \sum_{i = -n}^{\infty} a_i X^i\mbox{\@,} 
   \end{align}
for $n \in \Nat$ and $a_i \in \K$\@.
They may therefore be thought of as streams
which are preceded by a finite, possibly empty,
history, which may be used to "rewind computations“.
In fact, every formal Laurent series is of the form $X^{-n} \cdot f$, for some $n \in \Nat$ and for $f \in \Power{\K}{X}$ a formal power series.


\begin{remark} 
The rational functions in $\Rational{\K}{X}$ consist of
those formal Laurent series $\sum_{k=-n}^\infty a_k X^k$, 
for $k \in \Nat$, for which the sequence $(a_k)_{k\in\Nat}$ satisfies a linear recurrence 
relation.\footnote{
 That is, there is an $m \in \Nat$ and
 $c_0, \ldots, c_m \in \K$, not all zero, such that
 $c_0 a_n + \ldots + c_m a_{n + m} = 0$ for all $n\in \Nat$\@. 
} 
As a consequence, every formal Laurent series in $\Laurent{\F{q}}{X}$, for 
$\F{q}$ a finite field, is a
rational function in $\Rational{\F{q}}{X}$ 
if and only if its sequence $(a_i)_{i\in\Nat}$ of coefficients is eventually periodic.
The stream $(1, 1, 0, 1, 0, 0, 1, 0, 0, 0, 1, 0, 0, 0, 0, \ldots)$ therefore is 
not expressible as a rational function. 
\end{remark}
The value $\LI{f}$ of a formal Laurent series $f$ is the minimal index $k\in\Int$ with $f_k \neq 0$\@.
In this case, $f_k$ is the {\em lead coefficient} of $f$\@. 
%
Now, the set $\Laurent{\K}{X}$ of formal Laurent series is {\em orderable} (see Appendix~\ref{app:orderable}) by 
the {\em positive cone} $\PosCone{\Laurent{\K}{X}}$\@ of 
formal Laurent series with positive lead coefficient.
This set 
determines a strict ordering
$ f < g$, for $f, g \in \Laurent{\K}{X}$, which is defined to hold if and only if $g - f \in \PosCone{\Laurent{\K}{X}}$,
and a total ordering $f \leq g$,  which holds 
if and only if $f < g$ or $f = g$\@. 
\begin{proposition}\label{laurent-totally-ordererd}
$(\Laurent{\K}{X}, +, \cdot, \leq)$
is a totally ordered field.
\end{proposition}
As a consequence of Proposition~\ref{laurent-totally-ordererd}\@,
$\Laurent{\K}{X}$ is {\em formally real} 
($-1$ cannot be written as a sum of nonzero squares in $\Laurent{\K}{X}$)\@,
$\Laurent{\K}{X}$ is not algebraically closed (for example, the polynomial $X^2 + 1$ has no root)\@, and
$\Laurent{\K}{X}$ is of characteristic $0$\@
($0$ cannot be written as a sum of $1$s)\@.


\begin{figure}[t]
\begin{center}
    \begin{tikzcd}
     \Poly{\K}{X}    \arrow[r, hookrightarrow, "/"] \arrow[d, hookrightarrow, "*" ]
    & \Rational{\K}{X} \arrow[d, hookrightarrow, "*"]
    \\
      \Power{\K}{X}    \arrow[r, hookrightarrow, "/"]
    & \Laurent{\K}{X}
   \end{tikzcd}
   \end{center}
    \caption{Commuting Stream Embeddings ($\K$ is a field or at least an integral domain, $*$ denotes completion for valuation $|.|$, and $/$ the fraction field construction.}
    \label{fig:constructions}
\end{figure}


The Archimedean property (see~\cite{schechter1996handbook}) fails to hold
for $\Laurent{\K}{X}$, because 
$X  \not< 1 + 1 + \ldots + 1$, no matter how many
$1$’s we add together.
Analogously, $\nicefrac{1}{(1 + X)}$ is positive but {\em infinitesimal} over $\Laurent{\K}{X}$   
in the sense that $0 < \nicefrac{1}{(1 + X)} < [\nicefrac{1}{n}]$ for each $n \Defeq 1 + \ldots + 1$\@. 
As a non-Archimedean,  Cauchy complete, and totally order field,
$(\Laurent{\K}{X}, \leq)$ lacks the {\em least upper bound property}\@.
The total order $\leq$ on streams is also dense, and therefore it is not a linear continuum.  

The function $v: \Laurent{\K}{X} \to \Int \Union \{\infty \}$ with $v(0) \Defeq \infty$ and $v(f) \Defeq \LI{f}$, for $f \neq 0$, is a  (normalized) {\em valuation} on $\Laurent{\K}{X}$\@.
From the valuation $v$ one obtains,
using the convention $2^{\infty} \Defeq 0$,
the absolute value function
   $|.|: \Laurent{\K}{X} \to \Real^{\geq 0}$ 
by setting
   \begin{align}\label{def:absolute.value}
       |f| &\Defeq  2^{-v(f)}\mbox{\@.}
   \end{align}
By construction, $|.|$ is the {\em non-Archimedean absolute value} on $\Laurent{\K}{X}$
corresponding to the valuation $v$\@~\cite{neukirch2013algebraic}\@. 
%
The induced metric $d: \Laurent{\K}{X} \Cross \Laurent{K}{X} \to \Real^{\geq 0}$ with
    \begin{align}
        d(f, g) &\Defeq |f - g|
    \end{align}
measures the distance between $f$ and $g$ 
in terms of the longest common prefix. 
By construction,  
the {\em strong triangle inequality}
   $d(f, h) \leq \Max{d(f,g)}{d(g, h)}$\@.
holds for all $f, g, h \in \Laurent{\K}{X}$, 
and therefore $d$ is ultrametric.
\begin{proposition}\label{prop:ultrametric}
$(\Laurent{\K}{X}, d)$ is an ultrametric space. 
\end{proposition}

\begin{example}
The scaled identity function $I_f(x) \Defeq f \cdot x$, for $f \neq 0$, is
uniformly continuous in the topology induced by the metric $d$\@.\footnote{
The topology induced by the order $\leq$ on stream is identical to the topology induced by the prefix metric $d$\@. 
} 
For given $\varepsilon > 0$, let $\delta \Defeq \nicefrac{\varepsilon}{|f|}$\@.
Now,  $d(x, y) < \delta$ implies
$d(fx, fy) = |f|\,d(x, y) < |f|\delta = \varepsilon$ for all $x, y \in \K$\@.
\end{example}

\begin{proposition}\label{prop:add-mult-continuous}
Both addition and multiplication of formal Laurent series
in $\Laurent{\K}{X}$ are continuous in the topology 
induced by the prefix metric $d$\@. 
\end{proposition}
The notions of Cauchy sequences and convergence are 
defined as usual with respect to the metric $d$\@.
In particular, a sequence $(f_k)_{k\in\Nat}$ is {\em convergent in $\Laurent{\K}{X}$} provided that
       (1) there is an integer $J$ such that $J \leq \LI{f_k}$ for all $k \in \Nat$\@, and
       (2) for every $n \in \Nat$ there exists $K_n \in \Nat$ and $A_n \in \K$ such that if $k \geq K_n$ then $(f_k)_n = A_n$\@. 
The first condition says that there is a uniform lower bound for the indices of every entry $f_k$ of the sequence, whereas
the second condition says that if we focus attention on only the $n$–th power of $X$, and consider the sequence $(f_k)_n$ of coefficients of $X^n$ in $f_k$ as $k \to \infty$, then this sequence (of elements of $\K$) is eventually constant, with the ultimate value $A_n$\@. 
In this case,
        $f \Defeq \sum_{n = J}^\infty A_n X^n$
is a well-defined Laurent series, called the {\em limit} of the convergent sequence $(f_k)_{k\in\Nat}$\@. 
We use the notation $\lim_{k\to\infty} f_k = f$ to denote this relationship. 
For example, $\lim_{k\to\infty} X^k = 0$ and
$\lim_{k\to\infty} \sum_{i = 0}^k X^i = \nicefrac{1}{(1 - X)}$\@. 


Now, for a given sequence $(f_k)_{k\in\Nat}$ of formal
Laurent series, 
(1) the sequence $(f_k)$ is {\em Cauchy} iff $\lim_{k\to\infty} d(f_{k+1}, f_k) = 0$,
(2) the series $\sum_{k = 0}^\infty f_k \Defeq \lim_{K\to\infty} \sum_{k = 0}^K f_k$ converges iff $\lim_{k\to\infty} f_k = 0$, 
and 
(3) suppose that $\lim_{k\to\infty} f_k = f \neq 0$, then there exists an integer $N > 0$ 
such that for all $m \geq N$, $|f_m| = |f_N| = |f|$\@.
These properties follow directly from the fact
that $|.|$ is a non-Archimedean absolute value.


\begin{proposition}~\label{prop:laurent-complete}
  $(\Laurent{\K}{X}, d)$ is Cauchy complete. 
\end{proposition}
\begin{proof}
Let $(f_k)_{k \in \Nat}$ be a Cauchy sequence with $f_k \in \Laurent{\K}{X}$\@.
Then, by definition, for all $d\in\Nat$ there is $N_d \in \Nat$ such that $|f_n - f_m| < |X^d|$
for all  $n \geq m \geq N_d$\@.
But this means that $f_n - f_m \in X^d \cdot \Laurent{\K}{X}$\@.
Writing $f_k = \sum_{i \geq M_k} a_{k,i} X^i$ we get that $(a_{k,i})_{k \in \Nat}$ is constant for $k$ large enough, and there exists $J \in \Int$ such that
   $$
       \lim_{k\to\infty} f_k = \sum_{i \geq J 
       }
           (\lim_{k\to\infty} a_{k, i}) X^i \in \Laurent{\K}{X}\mbox{\@.}
   $$
\end{proof}
\noindent
Indeed, $\Laurent{\K}{X}$ is the Cauchy completion of $\Rational{\K}{X}$\@.
In summary, the stream embeddings discussed so far commute 
as displayed in Figure~\ref{fig:constructions}\@. 





\section{Real-Closedness}\label{sec:computable}


$\Laurent{\Real}{X}$ is an ordered field~(Proposition~\ref{laurent-totally-ordererd})\@.
Therefore, by definition (see Appendix~\ref{app:rcf}),
for demonstrating that $\Laurent{\Real}{X}$ is real-closed, we
still need to show the existence of a square root for streams and the existence of roots for all odd degree polynomials in $\Poly{\Laurent{\Real}{X}}{Y}$, 
where $Y$ is a single indeterminate.

The main step in this direction is an intermediate value property (IVP) for streams.
Recall that the standard proof of the intermediate value theorem for a continuous function over the field of real numbers is essentially based on
the fact that in the real field intervals and connected subsets coincide, and that continuous functions preserve connectedness.
When working with a disconnected ordered field, however, such an argument is not applicable anymore.

\begin{remark}
A non-Archimedean complete ordered field, such as $\Laurent{\Real}{X}$, lacks the least upper
bound property, and therefore also the dichotomic procedure for proving IVP\@. 
In this case not only the Archimedean proofs of IVP do not work, but IVP fails to hold in general. 
It is nevertheless true that IVP holds for polynomial and rational functions~\cite{bourbaki1964}\@. 
\end{remark}

 
\begin{lemma}[IVP]\label{lem:IVP}
For polynomial $P(Y) \in \Poly{\Power{\Real}{X}}{Y}$
and $\alpha, \beta \in \Power{\Real}{X}$ such that 
$P(\alpha) < 0 < P(\beta)$\@,
there exists $\gamma \in \Power{\Real}{X} \Inter (\alpha, \beta)$ with $P(\gamma) = 0$\@. 
\end{lemma}
\begin{proof}
Since $\Power{\Real}{X}$ is the Cauchy completion of $\Poly{\Real}{X}$, there are sequences $(a_n)_{n\in\Nat}$ and $(b_n)_{n\in\Nat}$ of polynomial streams with $a_n, b_n\in \Poly{\Real}{X}$ such that $\lim_{n\to\infty} a_n = \alpha$ and $\lim_{n\to\infty} b_n = \beta$\@. 
From the assumptions $P(\alpha) < 0 < P(\beta)$
and continuity of the polynomial $P$ in the topology induced by the prefix metric $d$, 
one can therefore find $a, b \in \Power{\Real}{X}$ in the sequences $(a_n)$ and $(b_n)$ with
$\alpha \leq a < b \leq \beta$ and
$P(a)  < 0 < P(b)$.
For continuity of $P$,  $P(\alpha) = P(\lim_{n\to\infty} a_n) = \lim_{n\to\infty} P(a_n)$\@.
Now, for $0 < \varepsilon \Defeq \nicefrac{|P(\alpha)|}{2}$\@, there exists $N \in \Nat$ such that for $d(P(a_n), P(\alpha)) < \varepsilon$ for all $n \geq N$\@.
In particular, $P(a) < 0$ for $a \Defeq a_N$\@.
The construction for $b$ is similar.

The proof proceeds along two cases.
If there is $\gamma \in \Power{\Real}{X} \Inter (a, b)$ such that $P(\gamma) = 0$ we are finished.
Otherwise, $f(\gamma) \neq 0$ for all $\gamma \in \Power{\Real}{X} \Inter (a, b)$\@.
    We define $\alpha_0 \Defeq a$, 
          $\beta_0 \Defeq b$,  and, for $m \in \Nat$, 
               \begin{align*}
               [\alpha_{m+1}, \beta_{m+1}]
                  &= \begin{cases}
                       [\alpha_m, \delta_m] : \text{ if } f(\delta_m) > 0 \\
                       [\delta_m, \beta_m] : \text{ if } f(\delta_m) < 0 \\
                  \end{cases}\mbox{\@,}
              \end{align*}
        where $\delta_m \Defeq \nicefrac{1}{2}(\alpha_m + \beta_m) \in\Power{\Real}{X} $\@.
        By assumption, $P(\delta_m) \neq 0$,
        and, by construction,
        $(\alpha_m)_{m\in\Nat}$ is
        a non-decreasing 
        and $(\beta_m)_{m\in\Nat}$ a non-increasing sequence in $\Poly{\Real}{X}$ such that, for all $m \in \Nat$,
        $\alpha_m < \beta_m$,
        $d(\alpha_m, \beta_m) \leq 2^{-m}$,
        $T(\alpha_m) < 0$, and $T(\beta_m) > 0$\@. 
     Therefore, both $(\alpha_m)_{m\in \Nat}$ and $(\beta_m)_{m\in \Nat}$ are Cauchy, 
      $(\alpha_m)_{m\in \Nat}$
     converges from below,
     and $(\beta_m)_{m\in \Nat}$ converges from above to a point $\gamma$\@.
     Now, $\gamma \in \Power{\Real}{X}$, since
     $\Power{\Real}{X}$ is the Cauchy completion of $\Poly{\Real}{X}$\@.
     Since $P$ is continuous we obtain
        \begin{align*}
             \lim_{m\to\infty} \underbrace{P(\alpha_m)}_{<0}
        =    P(\lim_{m\to\infty} \alpha_m) 
        =    P(\gamma) 
        =    P(\lim_{m\to\infty} \beta_m) 
        =     \lim_{m\to\infty} \underbrace{P(\beta_m)}_{>0}\mbox{\@,}
        \end{align*}
    and therefore $P(\gamma) = 0$\@.  
\end{proof}
A {\em real closed ring} is an ordered domain which has the 
intermediate value property for polynomials in one variable. 
From the IVP for formal power series we immediately get the following properties, which are characteristic for {\em real closed rings}~\cite{cherlin1983real}\@. 
\begin{proposition}\label{prop:char.rcr}~
 \begin{enumerate}
     \item\label{I} $f$ divides $g$ for all $f, g \in \Power{\Real}{X}$ with $ 0 < g < f$\@;
     \item\label{II} Every positive element in $\Power{\Real}{X}$ has a square root $\Power{\Real}{X}$;
     \item\label{III} Every monic polynomial in $\Poly{\Power{\Real}{X}}{Y}$ of odd degree has a root in $\Power{\Real}{X}$\@.
 \end{enumerate}
\end{proposition}
\begin{proof}
In each of the three cases we show that a certain polynomial changes sign, and hence
has a root. The relevant polynomials in $\Poly{\Power{\Real}{X}}{Y}$ are:
   \begin{enumerate}
       \item $f \cdot Y + g$ on $[0,1]$\@;
       \item $Y^2 - f$ on $[0, \Max{f}{1}]$\@;
       \item $Y^n + f_{n-1} \cdot Y^{n-1} + \ldots + f_1 \cdot Y + f_0$
             on $[-M, M]$, where $n\in \Nat$ is odd and
             $M \Defeq 1 + |f_{n-1}| + \ldots + |f_0|$\@. 
   \end{enumerate}
\end{proof}
\begin{example}
$\sqrt{(1, 1, 1, \ldots)} = (1, 2, 3, \ldots)$, since, using the identities in Example~\ref{ex:rats}, we get
    $
        (1, 1, 1, \ldots)^2
        = (\nicefrac{1}{(1 - X)})^2
        = \nicefrac{1}{(1 - X)^2}
        = (1, 2, 3, \ldots)$\@.
\end{example}
\begin{example}[Catalan]
The unique solution $f \in \Power{\Real}{X}$ of 
   $
       f = 1 + (X \cdot f^2)
   $
is obtained  
as
   $f = \nicefrac{2}{(1 + \sqrt{1 - 4X})}$\@,
which is the stream $(1, 1, 2, 5, 14, \ldots)$\@.
\end{example}

Alternatively, square roots of streams may be constructed as unique solutions of corecursive identities.
\begin{remark}[Corecursive definition of Square Root~\cite{rutten2001elements}]
Assume $f \in \Power{\Real}{X}$ with head coefficient $f_0 > 0$ and tail $f' \in \Power{\Real}{X}$\@.
Then, $\sqrt{f} \in \Power{\Real}{X}$ is the unique solution (for the unknown $g$)
of the corecursive identity 
     \begin{align}\label{def:sqrt}
     \Tl{g} &= \nicefrac{\Tl{f}}
                    {(\sqrt{f_0} + g)} 
     \end{align}
for the tail $g'$ of $g$, 
and the initial condition $g_0 = \sqrt{f_0}$ for the head $g_0$ of $g$\@.
Now, for all $f, g \in \Power{\Real}{X}$ with
$\Tl{f} > 0$\@,
if $g \cdot g  = f$ 
then either $g = \sqrt{f}$
or $g = -\sqrt{f}$,
depending on whether the head 
$g_0$
is positive or negative~(\cite{rutten2001elements}, Theorem 7.1)\@. 
\end{remark}

It is an immediate consequence of property~(\ref{I}) of Proposition~\ref{prop:char.rcr} that
the formal power series $\Power{\Real}{X}$ is a (proper) {\em valuation ring} of its field of fractions $\Laurent{\Real}{X}$\@; 
that is, $f$ or $f^{-1}$ lies in $\Power{\Real}{X}$ for each nonzero $f \in \Laurent{\Real}{X}$\@.
Since $\Power{\Real}{X}$ also satisfies the IVP~(Lemma~\ref{lem:IVP}) we get: 
\begin{corollary}\label{prop:real.closed.ring}
    $(\Power{\Real}{X}, +, \cdot, 0, 1, \leq)$ is a {\em real closed valuation ring}\@. 
\end{corollary}

Formal Laurent series, as the field of fractions of formal power series, inherit the properties ($\ref{II}$) and ($\ref{III}$) in Proposition~\ref{prop:char.rcr}\@. 

\begin{proposition}\label{prop:char.rcf}~
  \begin{enumerate}
 \item Every positive element in $\Laurent{\Real}{X}$ has a square root in $\Laurent{\Real}{X}$;
 \item\label{II'} Every monic polynomial in $\Poly{\Laurent{\Real}{X}}{Y}$ of odd degree has a root in $\Laurent{\Real}{X}$\@.
 \end{enumerate}
\end{proposition}
\begin{proof}
Assume $0 < \nicefrac{f}{g} \in \Laurent{\Real}{X}$\@. Then $0 < f\cdot g \in \Power{\Real}{X}$, and $\nicefrac{\sqrt{f\cdot g}}{g}$ is the square root of $\nicefrac{f}{g}$\@.
For establishing (\ref{II'}), assume
$P(Y) \in \Poly{\Laurent{\Real}{X}}{Y}$ be a polynomial of odd degree $n$\@. Choose $0 \neq h \in \Laurent{\Real}{X}$ such that $h \cdot P(Y) \in \Poly{\Power{\Real}{X}}{Y}$\@. 
Now, $Q(Y) \Defeq h^n \cdot P(\nicefrac{Y}{h})$ is a monic polynomial
in $\Poly{\Power{\Real}{X}}{Y}$ of odd degree.
Applying Proposition~(\ref{prop:char.rcr}.\ref{II}) to $q(Y)$ we see that
$p(Y)$ has a root in $\Laurent{\Real}{X}$\@. 
\end{proof}
It is an immediate consequence of Proposition~\ref{prop:char.rcf} that
the formal Laurent series are real closed (see Appendix~\ref{app:rcf})\@. 
\begin{corollary}\label{corr:laurent.rcf}
  $(\Laurent{\Real}{X}, +, \cdot, 0, 1, \leq)$ is a real closed (ordered) field.
\end{corollary}

\noindent
Therefore 
the ordering $\leq$ on $\Laurent{\Real}{X}$ is unique.

\section{Decision Procedure}\label{sec:qe}


The first-order theory $\TORCF$ of ordered real-closed fields (see Appendix~\ref{app:rcf}) admits 
quantifier elimination~\cite{tarski1998decision,cohen1969decision}\@.
That is, for every formula $\phi$ in the language $\LOR$ of ordered fields
there exists a 
quantifier free formula $\psi$  in this language 
with $\mathit{FV}(\psi) \subseteq \mathit{FV}(\phi)$\footnote{
$\mathit{FV}(.)$ denotes the set of free variables in a formula.}
such that $\TORCF \models (\phi \Iff \psi)$\@. 
Thus, Corollary~\ref{corr:laurent.rcf} implies quantifier elimination for the streams in $\Laurent{\Real}{X}$\@. 
\begin{theorem}~\label{thm:qe}
Let $\varphi$ be a first-order formula in the language $\LOR$ of ordered fields (and rings); 
then there is a computable function for deciding whether $\varphi$ holds in the
$\LOR$-structure $(\Laurent{\Real}{X}, +, \cdot, 0, 1; \leq)$ of streams.
\end{theorem}
Notice that decidability of $\Laurent{\Real}{X}$ already follows from the Corollary in~\cite{ax1965undecidability}, since the field $\Real$ is decidable and of characteristic $0$\@; this observation, however, does not yield quantifier elimination.  
As an immediate consequence of Theorem~\ref{thm:qe}, the structure of formal Laurent series with real-valued coefficients is {\em elementarily equivalent} to the real numbers in that they satisfy the same first-order $\LOR$-sentences. 

$\Power{\Real}{X}$ is a Henselian valuation ring, 
so that the theorems of Ax-Kochen and Ershov
apply, yielding completeness, decidability, and model completeness (in the language $\LOR \cup \{|\}$ of ordered rings augmented by the divisibility relation) for the theory $\RCVR$ of real closed valuation rings~(\cite{cherlin1983real}, Theorem 4A)\@.
In addition, there is an explicit quantifier elimination procedure for real closed valuation rings, which uses quantifier elimination on its fraction field as a subprocedure(\cite{cherlin1983real}, Section~2)\@. Therefore, by Corollary~\ref{prop:real.closed.ring},  we get a decision procedure
for first-order formulas for streams in $\Power{\Real}{X}$, which has quantifier elimination for $\Laurent{\Real}{X}$ as a subprocedure\@. 
\begin{theorem}~\label{thm:qe.ring}
Let $\varphi$ be a first-order formula in the language $\LOR \cup \{|\}$ of ordered rings extended with divisibility; 
then there is a computable function for deciding whether $\varphi$ holds in the
$\LOR \cup \{|\}$-structure $(\Power{\Real}{X}, +, \cdot, 0, 1; |, \leq)$ of streams.
\end{theorem}
Tarski's original algorithm for quantifier elimination has non-elementary 
computational complexity~\cite{tarski1998decision}, but
cylindrical algebraic decomposition provides
a decision procedures of complexity $d^{2^{O(n)}}$~\cite{collins1975quantifier}, 
where $n$ is the total number of variables (free and bound), 
and $d$ is the product of the degrees of the polynomials occurring in the formula. 
\noindent
\begin{theorem}\label{lem:qe.complexity}
Let $\varphi$ be a first-order formula in the language $\LOR$ of ordered fields.
Then the validity of $\varphi$ in the structure of streams
can be decided with complexity $d^{2^{O(n)}}$,
where $n$ is the total number of variables (free and bound), 
and $d$ is the product of the degrees of the polynomials occurring in $\varphi$\@. 
\end{theorem}
\noindent
This worst-case complexity is nearly optimal for quantifier elimination 
for real-closed fields~\cite{davenport1988real}\@. 
For existentially quantified 
conjunctions of literals
of 
the form
    $(\exists x_1, \ldots, x_k)
           \land_{i=1}^n p_i(x_1, \ldots, x_k) \bowtie 0$\@,
where $\bowtie$ stands for either $<$, $=$, or $>$ the 
worst-case complexity is $n^{k+1} \cdot d^{O(k)}$ arithmetic operations 
and polynomial space~\cite{basu1996combinatorial}\@.
Various implementations of decision procedures for real-closed fields 
are based on virtual term substitution~\cite{weispfenning1997quantifier} 
or on conflict-driven clause learning~\cite{jovanovic2013solving}\@. 

\section{Elementary Stream Calculus}\label{sec:extensions}

We consider definitional extensions of the first-order theory $\TORCR$ of ordered real-closed rings for encoding basic concepts of stream calculus.
Reconsider, for example, the stream circuit in Figure~\ref{fig:example.circuit}\@. 
The construction of a transfer function of this circuit,  as demonstrated in Example~\ref{ex:rational.stream.transformers}, 
is encoded as a first-order formula in the language $\LOR$ of (ordered) rings extended by a constant
symbol $\overline{X}$\@.
\begin{example}
    \begin{align*}
     & (\forall x, y, h_1, h_2, h_3)\, \\
     & ~~~~~~h_1 = \overline{X} \cdot h_2 \And
        h_3 = x + h_1 \And
        h_2 = h_3 \And
        y = x \\
     &~~~~~~~~~~~\Imp y = \overline{\nicefrac{1}{(1 - X)}} \cdot x\mbox{\@,}
    \end{align*}
\end{example}
whereby the logical variables 
are interpreted over streams in $\Power{\Real}{X}$\@. 
To obtain a decision procedure for these kinds of formula, we therefore
    \begin{itemize}
        \item Relativize quantification in $\TORCF$ to formal power series; 
        \item Define constant symbols $\overline{f}$ for rational streams $f$\@. 
    \end{itemize}

\paragraph{Streams.}  
As a direct consequence of Ax's construction~\cite{ax1965undecidability}, 
there is a
monadic formula with an $\exists\forall\exists\forall$ quantifier prefix and no parameters for uniformly defining the formal power series $\Power{\Real}{X}$ 
in $\Laurent{\Real}{X}$\@.\footnote{
This observation actually holds for any field of coefficients. 
}
Moreover, $\Power{\Real}{X}$ is 
$\forall\exists$-definable in $\Laurent{\Real}{X}$~(\cite{prestel2015definable}, Theorem~2 together with footnote~2), 
since the valuation ring $\Power{\Real}{X}$ is Henselian\@.
The model-theoretic developments in~\cite{prestel2015definable}, however, do
not provide an explicit definitional formula.
But explicit definitions of valuation rings in valued fields are studied in ~\cite{cluckers2013uniformly,anscombe2014existential,fehm2015existential}\@.
We therefore have a definition of a monadic predicate $\overline{S}(x)$ for the set of streams in 
$\Power{\Real}{X}$\@.
By relativization of quantifiers with respect to
this predicate $\overline{S}$ we therefore may assume from now on that all logical variables are interpreted over the streams in $\Power{\Real}{X}$\@. 

In addition, we are assuming definitions $\overline{R}(x)$ for given, and possibly finite, subsets $R$ of real number embeddings.  For example, the algebraic definition $(\forall x)\,\overline{\F{2}}(x) \Iff x = x^2$ characterizes the set
$\{[0],[1]\}$\@.

\COMMENT{
\paragraph{Reals.} We have already seen that a stream  $f \in \Power{\Real}{X}$ of the form $[r]$ embeds a real number $r$ as a stream.  
Let $f$ be of the form $(a_0, a_1, a_2, \ldots)$, then its square $f^2$ is of the form
   $$
      (a_0^2, 2a_0a_1, 2a_0a_2 + a_1^2,
            2a_0a_3 + 2a_1a_2, 
            2a_0a_4 + 2a_1a_3 + a_2^2, \ldots)\mbox{\@.}
   $$
In case $a_0 = 0$, we therefore get $f^2= 0$, and identification of $f$ with $f^2$ yields $f = 0$\@.
Otherwise, $a_0 \neq 0$\@. 
By identifying $f$ with $f^2$ and by component-wise comparison, we get (1) $a_0 = a_0^2$ or $-a_0 = a_0^2$, and (2) $a_i = 0$ for $i > 0$\@.
Therefore, the definition\footnote{
$\exists^1$ denotes the "there exists exactly one" quantification.
}
  \begin{align}
  (\forall x)\, \overline{R}(x) \Iff (x = x^2 \lor -x = x^2)\mbox{\@,}
  \end{align}
for $\overline{R}$ a fresh predicate symbol, yields the conservative extension $\TORCR[\overline{R}]$ of the theory $\TORCR$ of real closed rings
in the language $\LOR \cup \{\overline{R}\}$\@. 
As demonstrated above, $\overline{R}(x)$ holds in the structure $\Power{\Real}{X}$ if and only if $x$ is interpreted by a real number embedding $[r]$, for $r \in \Real$\@. 
}

\paragraph{Shift.} The {\em fundamental theorem of stream calculus}~\cite{rutten2019method} states that for every $f \in \Power{\Real}{X}$ there exist
unique $r \in \Real$ and $f' \in \Power{\Real}{X}$ with
       $f = [r] + X \cdot f'$\@.
In this case, $r$ is the {\em head} and $f'$ the {\em tail} of the stream $f$\@. 
Therefore, the definition 
    \begin{align}\label{X}
    (\forall z)\, \overline{X} = z
        \Iff (\forall y)\,(\exists^1 y_0, y')\, 
                  \overline{R}(y_0) \land y = y_0 + z \cdot y'\mbox{\@,}
    \end{align}
for $\overline{X}$ a fresh constant symbol, yields a conservative extension $\TORCR[\overline{S}, \overline{R}, \overline{X}]$ of the theory $\TORCR$, with $X$, as an element of $\Power{\Real}{X}$, the only possible interpretation
for the constant symbol $\overline{X}$\@. 
Notice that the definitional 
formula~(\ref{X}) for $\overline{X}$ requires two quantifier alternations.

\begin{example}
The basic stream constructors of stream circuits for addition $A$, multiplication $M_q$ by a rational $q$, and unit delay $D_1$
are defined by (the universal closures of)
   \begin{align*}
    \overline{A}(x_1, x_2) = y 
         &\Iff y = x_1 + x_2 \\
    \overline{M_{\nicefrac{n}{m}}}(x) = y 
         &\Iff m y = n x \\
    \overline{D_1}(x) = y 
         &\Iff y = \overline{X} \cdot x\mbox{\@,}
   \end{align*}
where $\overline{D_1}$, $\overline{A}$, and $\overline{M_{\nicefrac{n}{m}}}$ for $n, m \in \Nat$ with $m \neq 0$, are new function symbols,
and the variables are interpreted over $\Laurent{\Real}{X}$\@.
Synchronous composition of two stream circuits, say $S(x, y)$ and $T(y, z)$, is specified in terms of the quantified conjunction
    $
    \Exists{y}{S(x, y) \And T(y, z)}
    $\@,
where existential quantification is used for {\em hiding} the intermediate $y$ stream~\cite{srivas2005hardware}\@.
\end{example}

\paragraph{Rational Streams.}
We are now extending the language of ordered rings with constant symbols for rational streams.
This extended language is expressive, for example,
for encoding {\em equivalence} of rational stream transformers. 
We are considering rational streams
$f = \nicefrac{p(X)}{q(X)}$ with rational coefficients. In particular, the head for $q(X)$ is nonzero and $f \in \Power{\Real}{X}$\@.
Multiplication by $q(X)$ and by the least common multiple of the denominators of all rational coefficients in $p(X)$ and $q(X)$ yields an equality constraint in the language $\LOR[\overline{S}, \overline{R},\overline{X}]$\@.\footnote{
   $n \Defeq \underbrace{1 + \ldots +1}_{n\text{-times}}$
   for $n \in \Nat$\@. 
}
More precisely, 
let $\overline{\R_{\Q}}$ be a set of fresh constant symbols for all rational streams (except for $X$)
and $\TORCR[\overline{S}, \overline{R}, \overline{X}, \overline{\R_{\Q}}]$ the extension of $\TORCR$ by the definitions
   \begin{align}\label{def:rats}
   (\forall y)\,\overline{f} = y \Iff {\tilde{p}}(\overline{X}) \cdot y = {\tilde{q}}(\overline{X})
   \end{align}
for each (but $X$) rational stream $f$,
${\tilde{p}}(x) \Defeq c p(x)$, and
${\tilde{q}}(x) \Defeq c q(x)$\@, for $c$ the least common multiple of the denominators of coefficients of $p(x)$ and $q(x)$\@; then:  
$\TORCR[\overline{S}, \overline{R}, \overline{X}, \overline{\R_{\Q}}]$ is a conservative extension of $\TORCR$,
and all the symbols $\overline{f} \in \overline{\R_{\Q}}$ have the rational stream interpretation $f$\@.

\begin{remark}
Alternatively, a rational stream $f$ (with rational coefficients) can be finitely represented in terms of linear transformations
$H: \Q^d \to \Q$ and
$G: \Q^d \to \Q^d$, where $d$ is the finite dimension  of the linear span  
of the iterated tails of $f$~\cite{rutten2008rational}\@. 
Constraints for the finite number $d$ of linear independent iterated tails are obtained
from the anamorphism $\Ana{H, G}$, which is the unique homomorphism  from the 
coalgebra $\Pairing{H}{G} \in \Q^d \to \Q \Cross \Q^d$ to the corresponding final stream coalgebra\@. 
\end{remark}

\paragraph{Automatic Streams.}
We exemplify the encoding of a certain class of regular 
streams as (semi-)algebraic constraints in stream logic.
Consider the {\em Prouhet-Thue-Morse}~\cite{allouche2003automatic} stream $\mathit{ptm} \in \Power{\F{2}}{X}$,
for $\F{2}$ the finite field of characteristic $2$\@.
The $n^{\mathit{th}}$-coefficient of this stream is $1$ if and only 
if the number of $1$'s in the $2$-adic representation $[n]_2$ 
of $n$ is even. 
In other words, the $n^{\mathit{th}}$-coefficient is $1$ if and only if $[n]_2$ is in $0^*(10^*10^*)^*$\@. 
Such a stream is also said to be {\em automatic}~\cite{allouche2003automatic}\@.
Now, Christol's characterization~\cite{christol1980suites} of algebraic (over the rational functions with coefficients from a finite field)  power series in terms of deterministic finite automata with two states ("odd number of 1s", "even number of 1s")
implies that the stream $\mathit{ptm}$ is algebraic over $\Poly{\F{2}}{X}$\@.
In particular, the stream ${\mathit ptm}$ can be shown to be a root of the polynomial
   $X + (1 + X^2) \cdot Y + (1 + X)^3 \cdot Y^2$
of degree $2$  and its coefficients are in $\Poly{\F{2}}{X}$\@. 
 A semi-algebraic constraint for ruling out other than the intended solution(s) may be read-off, say, from of a Sturm chain. 

In this way, Christol's theorem supports the logical definition in stream logic of all kinds of analytic functions ($\sin$, $\cos$, \ldots) over finite fields. 
But not over the reals, as otherwise we could define the natural numbers using expressions such as $\sin{(\pi x)} = 0$\@.
And we could therefore encode 
undecidable identify problems 
over certain classes of analytic functions~\cite{richardson1994identity,laczkovich2003removal}\@. 

\paragraph{Heads and Tails.}
Based on the {\em fundamental law of stream calculus} for formal power series we define operators for stream projection and stream consing. 
Now, the theory $\TORCR[\overline{S}, \overline{R}, \overline{X}, \overline{hd}, \overline{tl}, \overline{cons}]$
 with the new (compared with $\TORCR[\overline{S}, \overline{R}, \overline{X}]$) definitional axioms
      \begin{align}
     (\forall x, x')\, \overline{tl}(x) = x' &\Iff (\exists x_0)\, \overline{R}(x_0) \land x = x_0 + \overline{X} \cdot x' \\
     (\forall x, x_0)\, \overline{hd}(x) = x_0 &\Iff \overline{R}(x_0) \land (\exists x')\, x = x_0 + \overline{X} \cdot x' \\
     (\forall x_0, x', y)\, \overline{cons}(x_0, x') = y &\Iff \overline{R}(x_0) \land y = x_0 + \overline{X} \cdot x'
      \end{align}
 is a conservative extension of $\TORCR$\@.
 Moreover, $\overline{hd}(x) = y$ ($\overline{tl}(x) = y$) holds in the structure $\Power{\Real}{X}$ if and only if $y$ is interpreted by the head (tail) of the interpretation of $x$\@; similarly for consing.
 
 With these definitions we may now also express corecursive identities in a decidable first-order equality theory.
The following example codifies the Fibonacci recurrence in our (extended) decidable logic.
\begin{example}
    \begin{align*}
        \overline{\mathit{hd}}(x) &= 0 \\
        \overline{\mathit{hd}}(\overline{\mathit{tl}}(x)) &= 1 \\
         \overline{\mathit{tl}}^2(x) 
           - \overline{\mathit{tl}}(x) - x   &= 0\mbox{\@.}
    \end{align*}
\end{example}
These {\em behavioral stream equations} 
are ubiquitous in stream calculus~\cite{rutten2019method}, for example, for specifying filter circuits. 
Filters in signal processing are usually specified as difference equations.
\begin{example}[$3$-$2$-filter] ~\label{example:3-2-filter}
A $3$-$2$-filter with input stream $x$ and output $y$ is specified in stream logic by three initial conditions and  the difference equation
   \begin{align*}
      \overline{hd}(y) &= 0 \\
     \overline{hd}(\overline{tl}(y)) &= 0\\
        \overline{hd}(\overline{tl}^2(y)) &= 0 \\
     \overline{tl}^3(y) &= c_0 x + 
             c_1 \overline{tl}(x) + \overline{tl}^3(x) + 
             c_2 c_3 \overline{tl}^2(y) + 
             c_4 \overline{tl}(y)\mbox{\@,}
   \end{align*}
for constants $c_0, \ldots, c_4 \in \Int$\@.
\end{example}

\begin{example}[Timing Diagrams]
The rising edge stream is specified in Scade-like~\cite{colacco2017scade} programming notation by means of the combined equation
      $$y = (0 \rightarrow x \And \lnot \mathit{pre}(x)\mbox{\@.}$$
  That is, the head of $y$ is $0$ and the tail of $y$ is specified by the expression to the right of the arrow. Hereby, $\mathit{pre}(x)$
  is similar to the shift operation in that
  $\mathit{pre}(x) = (\bot, x_0, x_1, \ldots)$, where $\bot$ indicates that the head element is undefined. 
  The rising edge stream is specified corecursively in stream logic by the identities
     \begin{align*}
        \overline{hd}(y) &= 0 \\
        \overline{tl}(y) &= \overline{\mathit{and}}(x, \overline{\mathit{not}}(\overline{tl}(x)))\mbox{\@,}
     \end{align*}
  for an arithmetic encoding of the (componentwise) logical stream operators 
  $\mathit{and}$ and $\mathit{not}$\@.
  \end{example}

\COMMENT{
\begin{example}[Timing Diagram]
  A binary signal is an element of $\Power{\F{2}}{X}$\@.  From a given binary signal $x$ we specify a binary signal $y$ which
  is $1$ iff there is a rising edge in $x$\@. 
  This relation is specified co-recursively by
     \begin{align*}
        \overline{hd}(x) &= 0 \\
        \overline{tl}(x) &= x \Overline{\And} \overline{\lnot} \overline{tl}(x)\mbox{\@,}
     \end{align*}
  where the characteristic predicate for binary signals and signal operators such as $\And$, $\lnot$ are defined corecursively in the obvious way.
  The rising edge stream is specified in Scade-like~\cite{colacco2017scade} programming notation by means of the combined equation
      $$y = 0 \rightarrow x \And \lnot \mathit{pre}(x)\mbox{\@.}$$
  That is, the head of $y$ is $0$ and the tail of $y$ is specified by the expression to the right of the arrow. Hereby, $\mathit{pre}(x)$
  is similar to the shift operation in that
  $\mathit{pre}(x) = (\bot, x_0, x_1, \ldots)$, where $\bot$ indicates that the head element is undefined.  
  \end{example}
}
Finally, the decision procedure for stream logic may be used in {\em coinduction} proofs for deciding whether or not a given binary stream relation is a bisimulation. 
\begin{example}[Bisimulation]~\label{ex:bisim}
A binary relation $B$ on streams, expressed as a formula in stream logic whith two free variables, is a {\em bisimulation}~\cite{rutten2019method} if and only if the 
$\LOR[\overline{S}, \overline{R}, \overline{X}, \overline{hd}, \overline{tl}]$ formula
   \begin{align}
   (\forall x, y)\, B(x, y) 
     \Implies \overline{hd}(x) = \overline{hd}(y)
     \And B(\overline{tl}(x), \overline{tl}(y))
   \end{align}
holds in the structure $\Power{\Real}{X}$ of streams. 
\end{example}

\section{Conclusions}\label{sec:conclusions}

First-order stream logic is expressive for encoding
elementary problems of stream calculus. 
It is decidable in doubly exponential time, and its decision procedure essentially is based on quantifier elimination for the theory of real closed ordered fields.
Some of the suggested encodings (relativization of quantifiers and the shift operator add new quantifier alternations), however, have the potential to substantially increase computation times.
It therefore still remains to be seen if and how exactly a quantifier 
elimination-based decision procedures for stream logic demonstrates practical 
advances, say, compared to the non-elementary
decision procedures for $MSO(\omega)$~\cite{klarlund2002mona,owre2000integrating}\@.
The goal hereby is to improve the quantifier elimination procedure specifically for stream logic encodings
as the basis for enlarging the scope of its practical applicability. 
It should also be interesting to characterize the set of definable corecursive stream functions in stream logic.  

The decision procedure for first-order stream logic may alternatively be based directly, that is, without relativization of the first-order quantifiers, on quantifier elimination for real closed valuation rings.
But these algorithms have not been studied and explored quite as well as quantifier elimination for real closed fields, and we are not aware of any computer implementation thereof. 

%
%
\bibliographystyle{splncs04}
\bibliography{stream}

\begin{thebibliography}{10}
\providecommand{\url}[1]{\texttt{#1}}
\providecommand{\urlprefix}{URL }
\providecommand{\doi}[1]{https://doi.org/#1}

\bibitem{allouche2003automatic}
Allouche, J.P., Shallit, J.: Automatic sequences: theory, applications,
  generalizations. Cambridge University Press (2003)

\bibitem{anscombe2014existential}
Anscombe, W., Koenigsmann, J.: An existential $\emptyset$-definition of
  ${F}_{q}[[t]]$ in ${F}_{q}((t))$. The Journal of Symbolic Logic
  \textbf{79}(4),  1336--1343 (2014)

\bibitem{ax1965undecidability}
Ax, J.: On the undecidability of power series fields. In: Proc. Amer. Math.
  Soc. vol. 16, no. 846, pp.~4--4 (1965)

\bibitem{basu1996combinatorial}
Basu, S., Pollack, R., Roy, M.F.: On the combinatorial and algebraic complexity
  of quantifier elimination. Journal of the ACM (JACM)  \textbf{43}(6),
  1002--1045 (1996)

\bibitem{bourbaki1964}
Bourbaki, N.: El\'ements de Math\'ematiques, vol. Livre II, Alg\`ebre, chap. 6,
  Groupes et corps ordonn\'es. Hermann, Paris (1964)

\bibitem{broy1986theory}
Broy, M.: A theory for nondeterminism, parallelism, communication, and
  concurrency. Theoretical Computer Science  \textbf{45},  1--61 (1986)

\bibitem{broy2023}
Broy, M.: Specification and verification of concurrent systems by causality and
  realizability. Theoretical Computer Science  \textbf{974}(114106),  1--61
  (2003)

\bibitem{broy2012specification}
Broy, M., St{\o}len, K.: Specification and development of interactive systems:
  focus on streams, interfaces, and refinement. Springer Science \& Business
  Media (2012)

\bibitem{buchi1969definability}
Buchi, J.R., Landweber, L.H.: Definability in the monadic second-order theory
  of successor1. The Journal of Symbolic Logic  \textbf{34}(2),  166--170
  (1969)

\bibitem{burge1975stream}
Burge, W.H.: Stream processing functions. IBM Journal of Research and
  Development  \textbf{19}(1),  12--25 (1975)

\bibitem{charalambides2018enumerative}
Charalambides, C.A.: Enumerative combinatorics. Chapman and Hall/CRC (2018)

\bibitem{cherlin1983real}
Cherlin, G., Dickmann, M.A.: Real closed rings ii. model theory. Annals of pure
  and applied logic  \textbf{25}(3),  213--231 (1983)

\bibitem{christol1980suites}
Christol, G., Kamae, T., Mend{\`e}s~France, M., Rauzy, G.: Suites
  alg{\'e}briques, automates et substitutions. Bulletin de la Soci{\'e}t{\'e}
  math{\'e}matique de France  \textbf{108},  401--419 (1980)

\bibitem{clarkson2010hyperproperties}
Clarkson, M.R., Schneider, F.B.: Hyperproperties. Journal of Computer Security
  \textbf{18}(6),  1157--1210 (2010)

\bibitem{cluckers2013uniformly}
Cluckers, R., Derakhshan, J., Leenknegt, E., Macintyre, A.: Uniformly defining
  valuation rings in henselian valued fields with finite or pseudo-finite
  residue fields. Annals of Pure and Applied Logic  \textbf{164}(12),
  1236--1246 (2013)

\bibitem{cohen1969decision}
Cohen, P.J.: Decision procedures for real and p-adic fields. Communications on
  pure and applied mathematics  \textbf{22}(2),  131--151 (1969)

\bibitem{colacco2017scade}
Cola{\c{c}}o, J.L., Pagano, B., Pouzet, M.: Scade 6: A formal language for
  embedded critical software development. In: 2017 International Symposium on
  Theoretical Aspects of Software Engineering (TASE). pp. 1--11. IEEE (2017)

\bibitem{collins1975quantifier}
Collins, G.E.: Quantifier elimination for real closed fields by cylindrical
  algebraic decomposition. In: Automata Theory and Formal Languages: 2nd GI
  Conference Kaiserslautern, May 20--23, 1975. pp. 134--183. Springer (1975)

\bibitem{courcelle2012graph}
Courcelle, B., Engelfriet, J.: Graph structure and monadic second-order logic:
  a language-theoretic approach, vol.~138. Cambridge University Press (2012)

\bibitem{cyrluk1997efficient}
Cyrluk, D., M{\"o}ller, O., Ruess, H.: An efficient decision procedure for the
  theory of fixed-sized bit-vectors. In: International Conference on Computer
  Aided Verification. pp. 60--71. Springer (1997)

\bibitem{davenport1988real}
Davenport, J.H., Heintz, J.: Real quantifier elimination is doubly exponential.
  Journal of Symbolic Computation  \textbf{5}(1-2),  29--35 (1988)

\bibitem{fehm2015existential}
Fehm, A.: Existential $\emptyset$-definability of henselian valuation rings.
  The Journal of Symbolic Logic  \textbf{80}(1),  301--307 (2015)

\bibitem{gilles1974semantics}
Gilles, K.: The semantics of a simple language for parallel programming.
  Information processing  \textbf{74},  471--475 (1974)

\bibitem{goguen1982security}
Goguen, J.A., Meseguer, J.: Security policies and security models. In: 1982
  IEEE Symposium on Security and Privacy. pp. 11--11. IEEE (1982)

\bibitem{jovanovic2013solving}
Jovanovi{\'c}, D., De~Moura, L.: Solving non-linear arithmetic. ACM
  Communications in Computer Algebra  \textbf{46}(3/4),  104--105 (2013)

\bibitem{kahn1976coroutines}
Kahn, G., MacQueen, D.: Coroutines and networks of parallel processes. Research
  Report  \textbf{INRIA-00306565} (1976)

\bibitem{klarlund2002mona}
Klarlund, N., M{\o}ller, A., Schwartzbach, M.I.: Mona implementation secrets.
  International Journal of Foundations of Computer Science  \textbf{13}(04),
  571--586 (2002)

\bibitem{laczkovich2003removal}
Laczkovich, M.: The removal of $\pi$ from some undecidable problems involving
  elementary functions. Proceedings of the American Mathematical Society
  \textbf{131}(7),  2235--2240 (2003)

\bibitem{mccullough1988noninterference}
McCullough, D.: Noninterference and the composability of security properties.
  In: Proceedings. 1988 IEEE Symposium on Security and Privacy. pp. 177--177.
  IEEE Computer Society (1988)

\bibitem{mclean1994general}
McLean, J.: A general theory of composition for trace sets closed under
  selective interleaving functions. In: Proceedings of 1994 IEEE Computer
  Society Symposium on Research in Security and Privacy. pp. 79--93. IEEE
  (1994)

\bibitem{moller1998solving}
M{\"o}ller, M.O., Ruess, H.: Solving bit-vector equations. In: International
  Conference on Formal Methods in Computer-Aided Design. pp. 36--48. Springer
  (1998)

\bibitem{neukirch2013algebraic}
Neukirch, J.: Algebraic number theory, vol.~322. Springer Science \& Business
  Media (2013)

\bibitem{niven1969formal}
Niven, I.: Formal power series. The American Mathematical Monthly
  \textbf{76}(8),  871--889 (1969)

\bibitem{owre2000integrating}
Owre, S., Ruess, H.: Integrating {WS1S} with {PVS}. In: International
  Conference on Computer Aided Verification. pp. 548--551. Springer (2000)

\bibitem{pradic2020some}
Pradic, P.: Some proof-theoretical approaches to Monadic Second-Order logic.
  Ph.D. thesis, Universit{\'e} de Lyon; Uniwersytet Warszawski. Wydzia{\l}
  Matematyki, Informatyki (2020)

\bibitem{prestel2015definable}
Prestel, A.: Definable henselian valuation rings. The Journal of Symbolic Logic
   \textbf{80}(4),  1260--1267 (2015)

\bibitem{richardson1994identity}
Richardson, D., Fitch, J.: The identity problem for elementary functions and
  constants. In: Proceedings of the international symposium on Symbolic and
  algebraic computation. pp. 285--290 (1994)

\bibitem{rutten2002streams}
Rutten, J.: On streams and coinduction. Tech. rep., CWI (2002)

\bibitem{rutten2019method}
Rutten, J.: The Method of Coalgebra: exercises in coinduction, vol. ISBN
  978-90-6196-568-8. CWI, Amsterdam (2019)

\bibitem{rutten2001elements}
Rutten, J.J.: Elements of stream calculus:(an extensive exercise in
  coinduction). Electronic Notes in Theoretical Computer Science  \textbf{45},
  358--423 (2001)

\bibitem{rutten2008rational}
Rutten, J.J.: Rational streams coalgebraically. Logical Methods in Computer
  Science  \textbf{4} (2008)

\bibitem{schechter1996handbook}
Schechter, E.: Handbook of Analysis and its Foundations. Academic Press (1996)

\bibitem{srivas2005hardware}
Srivas, M., Ruess, H., Cyrluk, D.: Hardware verification using {PVS}. In:
  Formal Hardware Verification: Methods and Systems in Comparison, pp.
  156--205. Springer (2005)

\bibitem{tarski1998decision}
Tarski, A.: A decision method for elementary algebra and geometry. Springer
  (1998)

\bibitem{van1966heidelberger}
van~der Waerden, B.: Algebra (1966)

\bibitem{weispfenning1997quantifier}
Weispfenning, V.: Quantifier elimination for real algebra—the quadratic case
  and beyond. Applicable Algebra in Engineering, Communication and Computing
  \textbf{8},  85--101 (1997)

\end{thebibliography}
%


\newpage
\appendix
\section{Orderable Fields}\label{app:orderable}

A field $\K$ is {\em orderable} if there exists a non-empty
$\PosCone{\K} \subset \K$ such that
\begin{enumerate}
    \item $0 \notin \PosCone{\K}$
    \item $(x + y), xy\in \PosCone{\K}$ for all $x, y \in \PosCone{\K}$
    \item Either $x \in \PosCone{\K}$ or $-x \in \PosCone{\K}$ for all $x \in \K \setminus \{0 \}$
\end{enumerate}
Provided that $\K$ is orderable we can generate a strict order on $\K$ by
$x < y$ if and only if $(y - x) \in \PosCone{\K}$\@.
Furthermore, a total ordering $\leq$ on $\K$ is defined by
$x \leq y$ if and only if $x < y$ or $x = y$,
and $(\K, \leq)$ is said to be a {\em (totally) ordered field}\@. 
Now, the {\em absolute value} of $x \in \K$ is defined 
by $|x| \Def \Max{-x}{x}$\@. 
The {\em triangle inequality}
    \begin{align}
        |x + y| &\leq |x| + |y|
    \end{align}
holds for ordered fields. 
As $-|x|-|y| \leq x + y \leq |x|+|y|$,
we have $|x+y|\leq |x|+|y|$, 
because $x+y\leq |x|+|y|$ and $-(x+y)\leq |x|+|y|$\@. 

Let $\K$ be an ordered field and $a \in \K \setminus \{0\}$ fixed.
The scaled identity function $I_a(x) \Defeq ax$ is
uniformly continuous in the order topology of $\K$.
For given $\varepsilon \in \PosCone{\K}$, let $\delta \Defeq \nicefrac{\varepsilon}{|a|}$\@.
Indeed, for all $x, y \in \K$, $|x-y| < \delta$ implies
$|ax - ay| = |a|\,|x - y| < |a|\delta = \varepsilon$\@.
Consequently, every polynomial in $\K$ is continuous.

A field $\K$ is orderable iff it is {\em formally real} (see~\cite{van1966heidelberger}, Chapter 11), that is,
$-1$ is not the sum of squares, or alternatively,
the equation $x_0^2 + \ldots + x_n^2 = 0$ has only trivial (that is, $x_k = 0$ for each $k$) solutions in $\K$\@.

\section{Real-Closed Fields}\label{app:rcf}

A field $\K$ is a {\em real closed field} if it satisfies the following.
  \begin{enumerate}
  \item $\K$ is formally real (or orderable).
  \item For all $x \in \K$ there exists $y \in \K$ such that
        $x = y^2$ or $x = -y^2$\@.
  \item For all polynomial $P \in \K[t]$ (over the single indeterminate $t$) with odd degree there exists $x \in \K$ such that $P(x) = 0$\@. 
  \end{enumerate}
Alternatively, a field $\K$ is {\em real closed} if $\K$ is 
formally real, but has no formally real proper algebraic extension field.

Let $\K$ be a real closed totally ordered field and $x \in \K$\@. 
Then $x > 0$ iff $x = y^2$ for some $y \in \K$\@.
Suppose $x > 0$, then, by definition of real-closedness, 
there exists $y \in \K$ such that $x = y^2$\@.
Conversely, suppose $x = y^2$ for some $y \in \K$, 
then, by the definition of $\PosCone{\K}$, we have $y^2 \in \PosCone{\K}$
for all $y \in \K$, and therefore $x > 0$\@.
Thus every real closed field is ordered in a unique way. 


Artin and Schreier’s theorem gives us two equivalent
conditions for a field $\K$ to be real closed: 
for a field $\K$ , the following are equivalent
   \begin{enumerate}
       \item $\K$ is real closed.
       \item $\K^2$ is a positive cone of $\K$ and every polynomial of odd degree has a root in $\K$\@.
       \item $\K(i)$ is algebraically closed and $\K \neq \K(i)$ (where $i$ denotes $\sqrt{-1}$)\@.
   \end{enumerate}
This characterization provides the basis (see axioms~\ref{ax:sqrt}) and~\ref{ax:root} below) for a 
first-order axiomatization of (ordered) real-closed fields.
The language of ordered rings (and fields),
$\LOR \Defeq \langle \leq \rangle; +, \cdot, -, 0, 1$
consists of a binary relation symbols $\leq$, 
two binary operator symbols, $+$, $\cdot$, 
one unary operator symbol $-$,
and two constant symbols $0$, $1$\@.
Now, the first-order theory $\TORCF$ of ordered real-closed fields consists of
all $\LOR$-structures $M$ satisfying the following set of axioms. 

\paragraph{Field Axioms.}
    \begin{enumerate}
    \item $(\forall x, y, z)\, x \cdot (y + z) = x \cdot y + x \cdot z$
    \item $(\forall x, y, z)\, x + (y + z) = (x + y) + z$
    \item $(\forall x, y, z)\, x \cdot (y \cdot z) = (x \cdot y) \cdot z$
    \item $(\forall x, y)\, x + y = y + x$
    \item $(\forall x, y)\, x \cdot y = y \cdot x$
    \item $(\forall x)\, x + 0 = x$ 
    \item $(\forall x)\, x + (-x) = 0$
    \item $(\forall x)\, x \cdot 1 = x$
    \item $(\forall x)\, x \neq 0 \Imp (\exists y)\, x \cdot y = 1$
    \end{enumerate}
\paragraph{Total Ordering Axioms.}
    \begin{enumerate}[resume]
    \item $(\forall x)\, x \leq x$
    \item $(\forall x, y, z)\, x \leq y \And y \leq z \Imp x \leq z$
    \item $(\forall x, y)\, x \leq y \And y \leq x \Imp x = y$
    \item $(\forall x, y)\, x \leq y \Or y \leq x$
    \item $(\forall x, y, z)\, x \leq y \Imp x + z \leq y + z$
    \item $(\forall x, y)\, 0 \leq x \And 0 \leq y \Imp 0 \leq x \cdot y$
    \end{enumerate}
\paragraph{Existence of Square Root.}
    \begin{enumerate}[resume]
    \item \label{ax:sqrt}
       $(\forall x)(\exists y)\, y \cdot y = x \Or y \cdot y = -x$
    \end{enumerate}
\paragraph{Every polynomial of odd degree has a root.}
    \begin{enumerate}[resume]
    \item \label{ax:root}
          $(\forall a_0, \ldots, a_n)\, a_n \neq 0 \Imp
                  (\exists x)\, a_0 + a_1 \cdot x + \ldots + a_n \cdot x^n = 0$ for odd $n \in \Nat$
    \end{enumerate}
If an $\LOR$-structure $M$ satisfies the axioms for ordered real-closed fields above, then $M$ is called a {\em model} of $\TORCF$\@. 
Any model of $\TORCF$ is {\em elementarily equivalent} to the real numbers. 
In other words, it has the same first-order properties as the field of ordered reals.

\end{document}